\documentclass[11pt]{article}

\usepackage[margin=1in]{geometry}
\usepackage{amsthm,amsmath,amsfonts,amssymb,mathrsfs}
\usepackage{graphicx}
\usepackage{tikz-cd}
\usepackage{multicol}
\usepackage{algorithm}
\usepackage{algorithmic}
\usepackage{subcaption}
\usepackage{rotating}
\usepackage{bbm}
\usepackage{booktabs}
\usepackage{xcolor}
\usepackage[numbers]{natbib}
\usepackage[colorlinks,citecolor=blue,urlcolor=blue]{hyperref}

\newcommand{\bbE}{\mathbb E}
\newcommand{\bbI}{\mathbb I }
\newcommand{\bbR}{\mathbb R}
\newcommand{\rev}[1]{\textcolor{black}{#1}}
\newcommand{\typo}[1]{#1}

\theoremstyle{plain}
\newtheorem{theorem}{Theorem}[section]
\newtheorem{lemma}[theorem]{Lemma}

\theoremstyle{definition}

\newtheorem{remark}[theorem]{Remark}

\title{Optimal scaling of MCMC algorithms: \color{black}the Hamiltonian approach}
\author{
\begin{tabular}{c}
P. Dobson\\
\small Maxwell Institute for Mathematical Sciences and Mathematics Department\\
\small Heriot-Watt University, Edinburgh, EH14 4AS, UK
\end{tabular}
\and
\begin{tabular}{c}
J. M. Sanz-Serna\\
\small Departamento de Matemáticas, Universidad Carlos III de Madrid\\
\small Avenida Universidad 30, 28911 Leganés, Madrid
\end{tabular}
\and
\begin{tabular}{c}
K. C. Zygalakis\\
\small Maxwell Institute for Mathematical Sciences and School of Mathematics\\
\small University of Edinburgh, Peter Guthrie Tait Rd, EH9 3FD, Edinburgh
\end{tabular}
}
\date{}

\begin{document}

\maketitle

\begin{abstract}
We present a simple, yet general \rev{approach} to \typo{study the scaling properties as the dimensionality of Metropolised MCMC sampling algorithms increases}. The study relies {\color{black}on the symmetries of the Hamiltonian formalism and} ultimately on \rev{the symmetry of the Metropolis-Hastings formula}. \rev{Our findings} contain, as particular cases, many known results for the Random Walk Metropolis, MALA  and other algorithms. In addition, they \typo{provide, in an easy way, new optimal scaling results} for a variety of proposal mechanisms, including implicit proposals and proposals generated with the help of differential equation integrators. The \rev{analysis} applies to targets that are products of a given, not necessarily univariate distribution, and also to cases where the different terms in the product are scaled differently. We show how to construct \rev{gradient-based MALA-like} proposals where the variance of the proposal as the dimension \(d\) increases may be taken as \(\mathcal{O}(1/d^\mu)\), with \(\mu>0\) arbitrarily small, to be compared with the values \(\mu = 1\) for Random Walk Metropolis and \(\mu=1/3\) for MALA.
\end{abstract}



\section{Introduction}
This paper presents a \rev{unified} approach to the investigation of the scaling properties of different MCMC sampling algorithms. Hundreds or perhaps thousands of MCMC algorithms based on Metropolisation have been suggested to sample from probability densities
\begin{equation}\label{eq:target}
\pi(q) \propto \exp(-V(q)), \qquad q\in\bbR^m.
\end{equation}
\rev{Two of the  best known proposals} are given by the Random Walk Metropolis (RWM) formula
\begin{equation}\label{eq:RWMproposal}
q^\star = q+\sqrt{\delta} p,\qquad p\sim N(0,I_m),
\end{equation}
and the MALA formula \cite{RT96}
\begin{equation}\label{eq:malaproposal}
q^\star = q - \frac{\delta}{2} \nabla V(q)+ \sqrt{\delta} p, \qquad p \sim N(0,I_m).
\end{equation}
The latter provides a consistent discretisation of the Langevin equation
\begin{equation}\label{eq:langevin}
dq(t) = -\frac{1}{2}\nabla V(q(t))\,dt+\, dW(t),
\end{equation}
which preserves the target \eqref{eq:target}.

A criterion to choose between all the \typo{different} algorithms \rev{is to study} how their parameters have to be varied as the dimensionality of the target increases. For targets consisting of \(d\) independent copies of a univariate distribution, it was proved in the pioneering contribution \cite{RGG97} that the variance of the RWM proposal should be scaled as \(\delta_d= \ell^2/d\), with a constant \(\ell\), to ensure that, as  \(d\rightarrow \infty\), the acceptance probability converges to a nontrivial limit not equal to zero or one. For MALA \cite{RR98}, the scaling is less demanding: \(\delta_d= \ell^2/d^{1/3}\). \rev{The best result available for a MALA-like sampler  corresponds to fMALA \cite{DRVZ17} where $\delta_{d}=\ell^2/d^{1/5}$, but this comes at the cost of calculating higher order derivatives of the potential \(V\)}. These studies also identify \emph{optimal} values of the acceptance probability, so that the algorithms are most efficient if \(\ell\) is tuned to achieve those optimal values.  It is \rev{remarkable that}, at least in the product of identical copies scenario, the optimal values of the acceptance probability may be proved to be independent of the target. The literature on optimal scaling is by now substantial; a selection of useful references include \cite{RR01,BS09,BPRSSS13,DRVZ17,B07,pillai2024optimal,PST12,NR06,YRR20,ZBK17,SR09,Bedard08,JLM14,JLM15} among others.

In this paper \typo{we provide} a simple, yet general approach to deriving scaling results that, on the one hand, makes it possible to recover the results in \cite{RGG97, RR98, BS09, DRVZ17}  and, on the other, provides, in an easy way, optimal scaling results for many other proposals. These include implicit proposals that need the solution of systems of nonlinear equations and proposals based on a variety of integrators taken from the Hamiltonian literature. Our approach allows  to easily construct \rev{MALA-like gradient-based} proposals with scaling \(\delta_d= \ell^2/d^\mu\) with arbitrarily small \(\mu>0\) in the i.i.d.\ product scenario.  The analysis applies to \emph{products} of a given, not necessarily univariate, distribution, and also to cases where the different components are \emph{scaled differently}. In the latter case, we investigate the effect of \emph{preconditioning} the algorithms.

The initial motivation for our research came when trying to study the scaling for the following implicit counterpart of \eqref{eq:malaproposal}
\begin{equation}\label{eq:maimlaproposal}
q^\star = q - \frac{\delta}{2} \nabla V\left(\frac{1}{2}(q^\star+q)\right)+ \sqrt{\delta} p, \qquad p \sim N(0,I_m),
\end{equation}
where the Langevin equation is discretised via the \emph{implicit midpoint rule}. For the corresponding Metropolis MCMC algorithm we use the acronym MAIMLA, Metropolis Adjusted Implicit Midpoint Algorithm. \typo{When compared with MALA, MAIMLA is of interest} in some circumstances see \cite{DSZinprep}: it is geometrically ergodic in a rather general setting and, in addition, regardless of the choice of \(\delta\), the acceptance rate is 100\% for Gaussian targets (and accordingly may be expected to be very high in the neighbourhood of a mode).

The key to the derivation of our  results is the   exploitation of  {\color{black} properties of the Hamiltonian formalism and the theory of numerical integrators for Hamiltonian systems. We initially focus on proposals that can be interpreted as approximations to Hamiltonian deterministic dynamics. For integrators that are volume preserving and reversible, the  symmetries of the Hamiltonian formalism \cite{sanz2018numerical} can be used to show that the \emph{average} of the negative logarithm of the acceptance probability is equal to \emph{half of its variance}, a result that goes back to \cite{BPRSSS13}. This relationship between the mean and variance is key  in identifying the specific form of the limiting acceptance rate. Remarkably, as proved by Vogrinc and Kendall \cite{VK21} (see also \cite{VLZ23}), the same relation  between mean and variance holds for general Metropolis-Hastings algorithms (i.e.\ to algorithms whose proposals cannot be interpreted within the Hamiltonian framework).}

The paper is structured in \rev{six} sections. In Section~\ref{sec:framework} we present a general \rev{setup} to view proposals as discretisations of deterministic Hamiltonian dynamics; we then  study the energy error, which determines the acceptance probability. The main scaling result \rev{in the Hamiltonian setting} is given in Section~\ref{sec:main}. \rev{That section} also contains  the analysis of how best to tune \rev{the} parameters to maximize the progress of the Markov chain. Also studied  is the
 effect of preconditioning both the dynamics and the algorithms in situations where the target includes components whose scales are widely
 different. \rev{Section~\ref{sec:optimal_MH} leaves the Hamiltonian scenario and {\color{black} briefly}
 discusses the general Metropolis-Hastings case}, while Section~\ref{sec:proof} contains the more technical parts of the proofs of the results. \rev{The final section} provides an outlook on how to use discretisations available in the Hamiltonian integration literature to obtain samplers with arbitrarily favourable scalings. A simple numerical illustration is also provided.
\section{Preliminaries}
\label{sec:framework}
\subsection{MCMC algorithms}
We introduce an auxiliary variable \(p\) (the momentum) and the extended target in \((q,p)\) space
\begin{equation}\label{eq:extendedtarget}
\Pi(q,p) \propto \exp(-H(q,p)),\qquad H(q,p) = \frac{1}{2}\| p\|^2+V(q),\qquad (q,p)\in\bbR^m\times\bbR^m.
\end{equation} (\(\|\cdot\|\) always denotes the standard Euclidean norm.)
The marginal on \(q\) is the target \eqref{eq:target} and \(p\sim N(0,I_m)\).

The Hamiltonian function \(H\) (energy) gives rise to the Hamiltonian differential equations
\begin{equation}\label{eq:hamiltonianequations}
\frac{d}{dt} q = p,\qquad\qquad \frac{d}{dt}p = -\nabla V(q).
\end{equation}
For each \(h>0\), the flow (solution map) \((q(0),p(0))\mapsto (q(h),p(h))\) of this system preserves the extended target \eqref{eq:extendedtarget}.

We consider MCMC algorithms with proposals generated with the help of one time step of a numerical integrator for \eqref{eq:hamiltonianequations}. The integrator is described by a transformation: \(\typo{\psi_h(q,p)}:\bbR^m\times\bbR^m \rightarrow \bbR^m\times\bbR^m\), parameterized by a steplength \(h>0\), such that for any solution of \eqref{eq:hamiltonianequations}, the value
\((q(h),p(h))\in\bbR^m\times\bbR^m\) \typo{may be approximated by}
\((q^\star,p^\star) = \typo{\psi_h(q(0),p(0))}\) in the sense that
\begin{equation}\label{eq:consistency}
   (q^\star,p^\star) -(q(h),p(h)) = \typo{o(h^\eta)}, \qquad h \rightarrow 0,
\end{equation}
for some \(\eta\geq 0\). Of course, the transformation \(\psi_h\) changes when the potential \(V\) is changed, but the dependence \typo{of \(\psi_h\)} on \(V\) is not incorporated to this notation; we will take this point up below.

We are interested in integrators with the following properties, that are discrete counterparts of properties of the Hamiltonian dynamics:

{\em Preservation of volume:}
\begin{equation*}
\forall (q,p)\in\bbR^m\times\bbR^m\qquad |{\rm det}(\psi_h^\prime(q,p))| = 1,
\end{equation*}

where \(\psi_h^\prime(q,p)\) is the Jacobian matrix of \(\psi_h(q,p)\) with respect to \((q,p)\).

{\em Reversibility:}
\begin{equation*}
\forall (q,p)\in\bbR^m\times\bbR^m \qquad \psi_h(q,p)= (q^\star,p^\star)   \Leftrightarrow  \psi_h(q^\star,-p^\star) = \typo{(q,-p)} .
\end{equation*}

Each transformation \(\psi_h\) gives rise to a \rev{Metropolis--Hastings} MCMC algorithm as follows. If \(q\) is the current state of \rev{the} Markov chain, draw \(p\sim N(0,I_m)\) \typo{(independent} from \(q\) and from previous history), compute \((q^\star,p^\star) =  \psi_h( q, p)\),  and take \(q^\star\) as the proposal.  If the integrator is volume preserving and reversible and the acceptance probability for the accept/reject mechanism is taken as
\begin{equation}\label{eq:acceptanceprobability}
a(q,p) = \min(1,\exp(-\Delta(q,p,h))), \qquad \Delta(q,p,h) = H(\textcolor{blue}{q^\star},p^\star)-H(q,p),
\end{equation}
then this procedure defines a Markov chain in \(\bbR^m\)  reversible with respect to the target \eqref{eq:target}.

Some examples:

\begin{itemize}

\item {\em Random Walk Metropolis (RWM):} The transformation \((q^\star,p^\star) = \psi_h(q,p)\) reads
\begin{equation*}
q^\star = q+hp,\qquad p^\star = p.
\end{equation*}
This is volume preserving and reversible. It satisfies \eqref{eq:consistency} \typo{with} \(\eta=0\), so that from a numerical analysis point of view the discretisation is not consistent with \eqref{eq:hamiltonianequations} (which is not surprising as it does not use the potential \(V\)). This lack of consistency poses no problem for the analysis below.

\item {\em MALA:} The transformation \((q^\star,p^\star) = \typo{\psi_h(q,p)}\) is given by the \typo{St\"{o}rmer--Verlet} integrator:
\begin{equation*}
p_{1/2} = p -\frac{h}{2}\nabla V(q), \qquad q^\star = q+hp_{1/2},\qquad p^\star = p_{1/2} -\frac{h}{2}\nabla V(q^\star).
\end{equation*}
It is trivial to show that \(\psi_h\) \typo{is volume preserving} and reversible; for sufficiently smooth potentials, property \eqref{eq:consistency} holds with \typo{\(\eta = 2\)} (one says that \typo{St\"{o}rmer--Verlet} is a consistent second order integrator).

After eliminating \(p_{1/2}\) and setting
\begin{equation}\label{eq:hdelta}
\delta = h^2
\end{equation}
the formula to compute the proposal may be written in the standard form \eqref{eq:malaproposal}.
 Note that \(\delta\) is the variance of the random term \(\sqrt{\delta} p\) in the definition of \rev{\(q^\star\)}.

\item{ \em MAIMLA:} \typo{Now}, \(\psi_h\) is defined by the implicit midpoint discretisation:
\begin{equation*}
q^\star -q = \frac{h}{2} (p^\star+p), \qquad p^\star-p = -h \nabla V\Big(\frac{1}{2}(q^\star+q)\Big).
\end{equation*}
Finding \(q^\star\) requires the solution of a system of \(m\) nonlinear equations. We assume that, for each \((q,p)\) the system has a unique solution for sufficiently small \(h\).
Again, it is trivial to check that the transformation is volume preserving and reversible and \eqref{eq:consistency} holds with \(\eta=2\) (second order).
\end{itemize}

Many other integrators are possible, as discussed in Section~\ref{sec:outlook}.
\subsection{Hamiltonian equations: properties of the energy error}
\label{sec:energyerror}


For solutions of \eqref{eq:hamiltonianequations}, \(H(q(t),p(t))\) does not vary with \(t\) (conservation of energy). Therefore, the value \(H(q,p)\) in \eqref{eq:acceptanceprobability} coincides with the value \(H(q(h),p(h))\), where \((q(t),p(t))\) denotes the  solution with initial value \((q(0),p(0))=(q,p)\) and thus \(\Delta(q,p,h)\) is the \emph{energy error} \(H(q^\star,p^\star)-H(q(h),p(h))\) resulting from approximating
the true value \((q(h),p(h))\) by means of the numerical solution \((q^\star,p^\star) =\psi_h(q,p)\).

We will
consider the following (mild) conditions on \(\Delta\).

{\em Condition 1.} There exist real functions \(\alpha(q,p)\), \(\rho(q,p,h)\) and a constant \(\nu>0\) such that
\begin{equation*}
\Delta(q,p,h) = h^{\nu}\alpha(q,p)+h^\nu\rho(q,p,h),\qquad \lim_{h\rightarrow 0}\rho(q,p,h) = 0.
\end{equation*}
A Taylor expansion of \(\Delta\) as a function of \(h\) shows that this condition will typically hold \typo{with} \(\nu = \eta+1\) (see \eqref{eq:consistency}) provided that \(V\) is \typo{sufficiently smooth}. To be precise, for the examples considered above one may easily prove:
\begin{itemize}
\item {\em RWM:} For continuously differentiable \(V\), \(\nu = 1\) and (superscripts denote Cartesian components of the vectors \(q\) and \(p\))
\begin{equation}\label{eq:alpharwm}
\alpha(q,p) = \sum_{i=1}^m \frac{\partial V}{\partial q^i}p^i.
\end{equation}

\item {\em MALA:} This has \(\nu=3\) and
\begin{equation}\label{eq:alphamala}
\alpha(q,p) = \frac{1}{4} \sum_{i,j=1}^m \frac{\partial^2 V}{\partial q^i \partial q^j}\frac{\partial V}{\partial q^i} p^j
-
\frac{1}{12} \sum_{i,j,k=1}^m \frac{\partial^3 V}{\partial q^i \partial q^j\partial q^k} p^ip^jp^k,
\end{equation}
provided that \(V\) is three times continuously differentiable.

\item {\em MAIMLA:} For \(V\) three times continuously differentiable,  \(\nu = 3\) and
\begin{equation}\label{eq:alphamaimla}
\alpha(q,p) =
\frac{1}{24}\sum_{i,j,k=1}^m \frac{\partial^3 V}{\partial q^i \partial q^j\partial q^k} p^ip^jp^k.
\end{equation}
\end{itemize}

The three integrators we have just described have \emph{odd} values of \(\nu\). This happens for all reversible algorithms \cite[Section 3.6.2]{sanz2018numerical}. Algorithms with \(\nu>3\) are discussed in Section~\ref{sec:outlook}.

The summations that feature in the formulas \eqref{eq:alpharwm}--\eqref{eq:alphamaimla} involving derivatives of the potential and components of the momentum are studied in the theory of symplectic integration, where they are called \emph{elementary Hamiltonians} \cite[Section~11.5.3]{sanz2018numerical}. Their structure may be easily described by using suitable graphs.

\rev{While Condition 1 describes the behaviour of \(\Delta\) as a function of \(h\), the following condition controls the variation of \(\Delta\) as a function of \((q,p)\).}

{\em Condition 2.} There exists a real function \(D(q,p)\) such that
\begin{equation*}
\rev{\sup_{0<h\leq 1} \frac{|\Delta(q,p,h)|^3}{h^{3\nu}} \leq D(q,p)}
\end{equation*}
and \(D\) is integrable with respect to the probability density \eqref{eq:extendedtarget}.

For typical integrators, this condition will hold in particular whenever \(\exp(-V(q))\) decays exponentially as \(\|q\|\rightarrow \infty\) and the derivatives of \(V\) grow at most polynomially. For RWM, MALA and MAIMLA, Taylor expansions in powers of \(h\)  with integral remainder show that it is sufficient that  \(V\) has polynomially growing continuous derivatives of order less than or equal to \(2, 4, 4\) respectively.

The properties of reversibility and volume preservation impose remarkable symmetry on the geometric behaviour of \(\psi_h\)
 (see \cite[Figure~6.1]{BRSS18}). That symmetry has important implications for
  the mean, second moment and variance of \(\Delta(q,p,h)\)
\begin{equation}\label{eq:mean}
\mu(h) = \bbE(\Delta(q,p,h)),\qquad
s(h) = \bbE(\Delta(q,p,h)^2),\qquad
\sigma^2(h)  =  s(h)-\mu(h)^2.
\end{equation}
 (Expectations \(\bbE\) are taken with respect to the extended target \eqref{eq:extendedtarget}. {\color{black} They are finite if Condition 2 holds.)}
 {\rev We have the following key result, which is essentially \cite[Theorem 7.1]{BRSS18}, a generalization of \cite[Proposition 3.4]{BPRSSS13}. These references operate with a  version of Condition 2 slightly weaker than that just considered. Here we have strengthened this hypothesis in order to avoid the somewhat obscure use of the dominated convergence theorem made in the proofs provided in those references.}

 \begin{theorem}\label{th:acta}Assume that the integrator is reversible, volume preserving and that Conditions 1 and 2 above hold. Then
 \begin{equation*}
 \lim_{h\rightarrow 0}\frac{\mu(h)}{h^{2\nu}}=\frac{\Sigma}{2},\qquad\qquad
 \lim_{h\rightarrow 0}\frac{\sigma^2(h)}{h^{2\nu}}=\Sigma
 \end{equation*}
 with
 \begin{equation*}
 \Sigma = \bbE(\alpha(q,p)^2).
 \end{equation*}
 \end{theorem}

{\color{black}The proof is postponed to Section~\ref{sec:proof}. There are however two very important points to be noticed at this point:}
\begin{itemize}
 \item While, according to Condition 1, \(\Delta(q,p,h)\) is \typo{of size} \(O(h^\nu)\) at each fixed \((q,p)\), its average is
 much smaller: \(O(h^{2\nu})\). Due to reversibility, if \((q^\star,p^\star) = \psi_h(q,p)\), then (see \cite[Figure~6.1]{BRSS18})
 \begin{equation}\label{eq:opposite}
 \Delta(q,p) = -\Delta(q^\star,-p^\star),
 \end{equation}
 a relation which leads to much cancellation in the computation of \(\bbE(\Delta)\).  As we shall see later this relation corresponds to the symmetry inherent in the Metropolis-Hastings formula or, ultimately, in the definition of detailed balance.
 \item In addition, for \(h\) small, the \emph{mean of \(\Delta\) is approximately half its variance}. This relation between mean and variance, valid for all volume preserving, reversible integrators, will explain the common expression for the limit acceptance probability they share.
 \end{itemize}

 \subsection{Scaling the potential}\label{sec:scaling}

We now investigate the effect of scaling the potential. Let \(\lambda>0\) be a scaling factor, and consider the scaled potential \(V^\lambda(q) = V(q/\lambda)\), where \(V\) is as above. We shall associate the superindex \(\lambda\) with objects pertaining to the scaled potential. The scaled probability density is \(\pi^\lambda(q) \propto \exp( -V^\lambda(q))\), which leads to the scaled Hamiltonian \(H^\lambda(q,p)= (1/2)\|p\|^2+V(q/\lambda)\) with Hamiltonian equations
\begin{equation*}
\frac{d}{dt} q = p,\qquad\qquad \frac{d}{dt}p = -\frac{1}{\lambda}\nabla V(q/\lambda).
\end{equation*}

After the change of variables
\begin{equation*}
\bar{q} = q/\lambda,\qquad \bar{p}= p,\qquad \bar{t} = t/\lambda
\end{equation*}
the scaled equations of motion become
\begin{equation*}
\frac{d}{d\bar{t}} \bar{q} = \bar{p},\qquad\qquad \frac{d}{d\bar{t}}\bar{p} = -\nabla V(\bar{q}),
\end{equation*}
a system that coincides with the original, unscaled \eqref{eq:hamiltonianequations}. Thus, for the differential equations, the effect of scaling the potential may be reproduced by using the unscaled dynamics as applied with a scaled time and scaled dependent variables. We now demand that the same happens at the level of the integrator, i.e.\
if \(\psi_h(q,p)\) is the transformation associated with the original potential \(V\) and \(\psi_h^\lambda(q,p)\) the transformation when the integrator is used for the rescaled differential equations:
\begin{equation}\label{eq:equivariance}
(q^\star,p^\star) = \psi_h^\lambda(q,p)\Leftrightarrow (q^\star/\lambda,p^\star) = \psi_{h/\lambda}(q/\lambda,p).
\end{equation}
We will say that an integrator is \emph{equivariant} if the two following processes lead to the same result: (i) rescaling the variables in the differential system and then applying the integrator, (ii) applying the integrator to the original differential equations and then rescaling variables in the discrete equations. The notion of equivariance is akin to the notion of dimensional correctness in physics. All integrators of any interest are equivariant; the performance of a nonequivariant integrator would depend on the particular units used to measure the quantities involved.

For an equivariant integrator the energy error \(\Delta^\lambda(q,p,h)\) for the scaled potential  satisfies
\begin{equation*}
\Delta^\lambda(q,p,h)= \Delta(q/\lambda,p,h/\lambda),
\end{equation*}
where \(\Delta\) is the energy error for the original potential. Furthermore we have the following elementary result:
\begin{lemma}\label{lem:changein expectation} Let \(F(q,p)\) be a function defined in \(\bbR^m\times\bbR^m\) and \(\lambda>0\). The expectation of \(F(q/\lambda,p)\) when
\((q,p)\) is distributed according to the density \(\propto \exp(-H^\lambda(q,p))\) coincides with the expectation of \(F(q,p)\) when
\((q,p)\) has density \(\propto \exp(-H(q,p))\).
\end{lemma}
\begin{proof}The first expectation is:
\begin{equation*}
\frac{\int_{\bbR^m\times \bbR^m}F(q/\lambda,p)\:\exp(-((1/2)\|p\|^2+V^\lambda(q)))\,dqdp}
{\int_{\bbR^m\times \bbR^m}\exp(-((1/2)\|p\|^2+V^\lambda(q)))\,dqdp}.
\end{equation*}
After taking \(q/\lambda\) as a new variable in both integrals, the quotient is seen to coincide with the second expectation.
\end{proof}

In addition,  with the notation
\begin{equation}\label{eq:momentsoldandnew}
\mu^\lambda(h) = \bbE^\lambda(\Delta^\lambda(q,p,h)),\qquad
s^\lambda(h) = \bbE^\lambda(\Delta^\lambda(q,p,h)^2),\qquad
\sigma^{2,\lambda}(h)  =  s^\lambda(h)-\mu^\lambda(h)^2,
\end{equation}
where \(\bbE^\lambda\) refers to expectations with respect to the density \(\Pi^\lambda(q,p)\propto\exp(-H^\lambda(q,p))\). Equation  \eqref{eq:equivariance} and Lemma~\ref{lem:changein expectation} show that, for an equivariant integrator,
\begin{equation}\label{eq:scaledmoments}
\mu^\lambda(h) = \mu(h/\lambda),\qquad s^\lambda(h)=s(h/\lambda), \qquad \sigma^{2,\lambda}(h) = \sigma^2(h/\lambda),
\end{equation}
where the quantities in the right hand-sides refer to the unscaled potential  and were defined in \eqref{eq:mean}.

\section{Optimal scaling using the Hamiltonian framework}\label{sec:main}
For \(d\geq 1\), we set \(N=dm\) and write vectors in \(\bbR^N\) as \(Q=(q_1,\dots,q_d)\) with \(q_i\in\bbR^m\).
We  fix the potential \(V\) in \(\bbR^m\) as in \eqref{eq:target}  and consider the target probability density in \(\bbR^N\) defined by
\begin{equation}\label{eq:product target}
\pi^\lambda(Q)\propto \prod_{i=1}^d \exp(-V(q_i/\lambda_i)),
\end{equation}
where \(\lambda_i>0\) are scaling constants. After defining the potential
\begin{equation*}
{\cal V}(Q) = {\cal V}(q_1,\dots,q_d)  = \sum_{i=1}^d V(q_i/\lambda_i),
\end{equation*}
the target is of the form \(\pi^\lambda(Q)\propto \exp(-{\cal V}(Q))\) and the algorithms in Section~\ref{sec:framework} are applicable. The Hamiltonian
that features in the extended target \(\propto \exp(-{\cal H}(Q,P))\) is
\begin{equation*}
{\cal H}(Q,P) = {\cal H}(q_1,\dots,q_d,p_1\dots,p_d) = \frac{1}{2} \sum_{i=1}^d  \|p_i\|^2+\sum_{i=1}^d V(q_i/\lambda_i),
\end{equation*}
and the Hamiltonian equations  in \(\bbR^N\times\bbR^N\) are  obtained by simply juxtaposing the \(d\) Hamiltonian systems  in \(\bbR^m\times\bbR^m\) corresponding to the different
\((q_i,p_i)\). There is no coupling between the  differential equations for the different components and the same is true for the numerical integrator.
However in the accept/reject mechanism the components come together because the acceptance probability is \(\min(1,\exp(-\Delta_d))\) with
(stars denote proposals)
\begin{equation*}
\Delta_d (Q,P,h)  = {\cal H}(q_1^\star,\dots,q_d^\star,p_1^\star,\dots,p_d^\star) - {\cal H}(q_1,\dots,q_d,p_1\dots,p_d).
\end{equation*}
With the notation as in the previous section,
\begin{equation*}
\Delta_d (Q,P,h) =  \sum_{i=1}^d \Delta^{\lambda_i}(q_i,p_i,h).
\end{equation*}

\subsection{A general scaling result}
We are interested in finding conditions that ensure that  \(\Delta_d\) has a distributional limit as \(d\rightarrow \infty\).  It is reasonable to assume that the time-step \(h\) in the integration has to shrink as \(d\) increases, so that the larger number of terms being summed in the preceding display is compensated by a decrease in the size of the individual terms. We denote by \(h_d\) the integration time-step to be used when there are \(d\) components.

Here is the main result of this paper:

\begin{theorem}\label{th:main}In the setup just described, assume that the integrator being used is equivariant, volume preserving and reversible and that, for the potential \(V\) in \(\bbR^m\), Conditions 1 and 2 hold. Assume that
  the scaling constants \(\lambda_i\), i = 1, 2,\dots, satisfy
\begin{equation}\label{eq:gamma}
\lim_{d\rightarrow \infty} d^{-\gamma} \sum_{i=1}^d \frac{1}{\lambda_i^{2\nu}} = K <\infty
\end{equation}
for some constants \(\gamma>0\), \(K>0\) (\(\nu\) is as in Condition 1) and, in addition,
\begin{equation}\label{eq:conditionmax}
 \lim_{d\rightarrow \infty} d^{-\gamma} \max_{i =1,\dots,d} \frac{1}{\lambda_i^{2\nu}} = 0.
\end{equation}
Choose \(h_d = \ell/d^{\gamma/(2\nu)}\), for some constant \(\ell>0\).
Then, for the density \(\propto \exp(-{\cal H})\)  in \(\bbR^N\times\bbR^N\), the expectation \(\bbE(a_d)\) of the acceptance probability \[a_d(Q,P,h)= \min(1,\exp(-\Delta_d(Q,P,h_ d)))\] converges as \(d\rightarrow \infty\) to the limit
\begin{equation}\label{eq:main}
A = 2 \Phi\left(-\frac{\ell^{\nu}\sqrt{K\Sigma}}{2}\right),
\end{equation}
where \(\Phi\) denotes the cumulative distribution function of the standard normal distribution \(N(0,1)\) and \(\Sigma\) is defined in Theorem~\ref{th:acta}.
\end{theorem}

A simple outline of the proof of this result will be given presently; the missing technical details are provided in Section~\ref{sec:proof} below. Since the integrator is assumed to be equivariant, according to \eqref{eq:equivariance},
\begin{equation}\label{eq:Deltad}
\Delta_d = \sum_{i=1}^d \Delta(q_i/\lambda_i,p_i,h/\lambda_i).
\end{equation}
We will use a Central Limit Theorem to identify the distributional limit of \(\Delta_d\) as \(d\rightarrow \infty\). We consider the triangular array of random variables whose \(d\)-row, \(d=1,2,\dots\), is
\begin{equation}\label{eq:row}
\Delta(q_1/\lambda_1,p_1,h_d/\lambda_1), \Delta(q_2/\lambda_2,p_2,h_d/\lambda_2),\dots, \Delta(q_d/\lambda_d,p_d,h_d/\lambda_d).
\end{equation}
Each \(q_i\) is distributed as \(\propto \exp(-V(q_i/\lambda_i))\) and each \(p_i\) is normal with zero mean and covariance matrix \(I_m\); all these variables are mutually independent. According to \eqref{eq:momentsoldandnew},  the variances of the variables in \eqref{eq:row} are
\begin{equation*}
\sigma^2(h_d/\lambda_1), \sigma^2(h_d/\lambda_2),\dots, \sigma^2(h_d/\lambda_d),
\end{equation*}
and by  Theorem~\ref{th:acta}, the variance of \(\Delta_d\) will  approximately be
\(
\sum_{i=1}^d {h_d^{2\nu}}/{\lambda_i^{2\nu}} \Sigma
\)
;
the corresponding expectation will be approximately  \emph{a half} of this value. Using the definition of \(h_d\):
\begin{equation*}
\sum_{i=1}^d \frac{h_d^{2\nu}}{\lambda_i^{2\nu}} \Sigma = \ell^{2\nu} \: d^{-\gamma} \sum_{i=1}^d \frac{1}{\lambda_i^{2\nu}}\Sigma,
\end{equation*}
a quantity that, according to the hypotheses on the \(\lambda_i\), tends to \(\ell^{2\nu} K \Sigma\) as \(d\rightarrow \infty\). By invoking the Central Limit Theorem for triangular arrays (\eqref{eq:conditionmax} is required to check the Lindeberg condition), we conclude that \(\Delta_d\) converges in distribution to a normal variable \(\Delta_\infty\) of mean \((1/2) \ell^{2\nu}K\Sigma\) and variance \(\ell^{2\nu}K\Sigma\). From the boundedness of \(x\mapsto \min(1,\exp(-u))\), \(\bbE(a_d)\) converges to \(\bbE(\min(1,\exp(-\Delta_\infty)))\); the last expectation may be computed analytically and has the value \eqref{eq:main}. This concludes the outline of the proof.

Here are some useful choices of scaling constants:
\begin{enumerate}
\item When \(\lambda_i=1\) for each \(i\) (the different \(q_i\) share the common density \(\propto\exp(-V)\)), \(\gamma=1\), \(K=1\) (\eqref{eq:conditionmax} trivially holds) and \(h_d = \ell/d^{1/(2\nu)}\).

\item For \(\lambda_i = 1/i^\kappa\), with \(\kappa>0\),
\begin{equation*}
\sum_{i=1}^d \frac{1}{\lambda_i^{2\nu}} = \sum_{i=1}^d i^{2\nu\kappa} \sim \frac{d^{2\nu\kappa+1}}{2\nu\kappa+1},\qquad d\rightarrow \infty,
\end{equation*}
and therefore \(\gamma =2\nu\kappa+1 \), \(K = 1/(2\nu\kappa+1)\) leading to \(h_d = \ell/d^{\kappa+1/(2\nu)}\). Again \eqref{eq:conditionmax} holds.
\end{enumerate}

Note that for choice 2 the exponent of \(d\) in the denominator of \(h_d\) exceeds the exponent \(1/(2\nu)\) of the \(\lambda_i = 1\) case by an amount \(\kappa\). In fact, for choice 2, the time step \(h_d\) has to be selected to cater for the small \(\lambda_i\) components for which \(q_i\) undergoes fast changes. As a result the components \(q_i\) with large \(\lambda_i\) are integrated with a time step that is too small relative to their own rate of change. The need to use very small values of \(h_d\) may be circumvented by using preconditioning (see Section~\ref{sec:preconditioning}). On the other hand, if \(\kappa\) increases, then \(K = 1/(2\nu\kappa+1)\) decreases and the acceptance rate in \eqref{eq:main} increases; this was of course to be expected because \(h_d = \ell/d^{\kappa+1/(2\nu)}\) is a decreasing function of \(\kappa\).

Let us finish this section by applying Theorem~\ref{th:main} to the algorithms considered in Section~\ref{sec:framework}:
\begin{itemize}
\item \emph{RWM:} This has \(\nu=1\). For the i.i.d. case with \(\lambda_i=1\), where \(\gamma=1\), \(K=1\), the scaling is
\(h_d = \ell/\sqrt{d}\) (the variance of the proposal  will then be \(\delta_d=\ell^2/d\)).
Under this scaling, the limit expected acceptance rate will be
\begin{equation*}
A =  2 \Phi\left(-\frac{\ell\sqrt{\Sigma}}{2}\right)
\end{equation*}
with (see \eqref{eq:alpharwm})
\begin{equation*}
\Sigma = \bbE\big( \left(\sum_{i=1}^m \frac{\partial V}{\partial q^i}p^i\right)^2 \big) =
\bbE\big( \sum_{i=1}^m \left(\frac{\partial V}{\partial q^i}\right)^2 \big)=\bbE\big(\|\nabla V(q)\|^2\big),
\end{equation*}
where we have taken into account that the \(p^i\) are uncorrelated and have unit variance.
In the particular case \(m=1\), the scaling \(\delta_d = \ell^2/d\) and the formula for \(A\) were first presented in the pioneering paper \cite{RGG97}.
\item\emph{MALA:} This has \(\nu=3\). For the i.i.d. case with \(\lambda_i=1\), where \(\gamma=1\), \(K=1\), the scaling is therefore
\(h_d = \ell/d^{1/6}\), i.e.\ \(\delta_d = \ell^2/d^{1/3}\). Under this scaling, the limit expected acceptance rate is
\begin{equation*}
 A = 2 \Phi\left(-\frac{\ell^3\sqrt{\Sigma}}{2}\right),
\end{equation*}
where \(\Sigma = \bbE(\alpha^2)\) with \(\alpha\) given in \eqref{eq:alphamala}.

 For the particular case \(m=1\),  the scaling of \(\delta_d\) and the formula for \(A\) were first derived in \cite[Theorem~1]{RR98} (see also \cite{RR01}). In that particular case, noting that \(p\sim N(0,1)\) and integrating by parts in the integral for the expectation,  one may show that
\begin{equation*}
\Sigma = \frac{1}{48} \bbE\big(5V^{\prime\prime\prime}(q)^2+3 V^{\prime\prime}(q)^3\big).
\end{equation*}

For the scaling \(\lambda_i = 1/i^\kappa\), the theorem gives  \(h_d = \ell/d^{\kappa+1/6}\) or \(\delta_d = \ell^2/d^{2\kappa+1/3}\), a result that, for the particular case \(m=1\), was derived in \cite[Theorem~5.2]{BS09}; this reference does not identify the value of \(A\), which, as we shall discuss later,  is required to derive optimal scaling results and was given in \cite{BRS09}.
\item\emph{MAIMLA:} The only difference with MALA is in the value of \(\Sigma\). A comparison between \eqref{eq:alphamala} and \eqref{eq:alphamaimla} suggests that, in general, \(|\alpha_{MALA}|> |\alpha_{MAIMLA}|\) so that MAIMLA enjoys larger acceptance rates for a given  steplength. Numerical experiments confirm this \cite{DSZinprep}.

\end{itemize}

\subsection{Squared Jumping Distance}

According to \eqref{eq:main}, the expected acceptance rate approaches \(100\%\) as \(\ell\rightarrow 0\). However very small values of \(\ell\) entail proposals that are too close to the current state and result in large correlations in the Markov chain. To analyse the optimal choice of \(\ell\), it is useful to consider the squared jumping distance defined as
\begin{equation}\label{eq:sjd}
{\cal SJD}_d = \frac{1}{d} \sum_{i=1}^d \bbE(\|q_i^\star-q_i\|^2 \bbI_{\{a_d(Q,P,h_d) \leq U\}}),\qquad U\sim{\rm Uniform}(0,1);
\end{equation}
here, a star denotes  proposal and the indicator function determines whether the proposal is accepted or otherwise. Note the averaging over the different \(q_i\).

For the RWM algorithm applied to the target \eqref{eq:target}, \rev{\(q^\star-q = h p\)}. For algorithms (including MALA and MAIMLA) based on integrators consistent with the Hamiltonian system \eqref{eq:hamiltonianequations},
\rev{\(q^\star-q = h p+ o(h)\)}. More generally, we introduce the following condition:

\emph{Condition 3.} When the integrator is applied to the target \eqref{eq:target}, as \(h\rightarrow 0\),
\begin{equation*}
r(h): = \bbE\left(\|\frac{1}{h}(q^\star-q)\|^2-\|p\|^2 \right) \rightarrow 0.
\end{equation*}
(Here \(q \sim \exp(-V(q))\) and \(p \sim N(0,I_m)\).)

This condition is easily checked by Taylor expansion of \rev{\(q^\star\)} as a function of \(h\). In particular, it will \typo{hold} if \(\nabla V(q)\) grows at most as a polynomial and \(\exp(-V(q))\) has tails that decrease exponentially.

The next result identifies the limit of the squared jumping distance:
\begin{theorem}\label{th:sjd} Under the hypotheses of Theorem~\ref{th:main}, assume in addition that Condition 3 is satisfied. Then
\begin{equation*}
\lim_{d\rightarrow \infty}d^{\gamma/\nu} {\cal SJD}_d = m\ell^2 A= 2m\ell^2 \Phi\left(-\frac{\ell^{\nu}\sqrt{K\Sigma}}{2}\right).
\end{equation*}
\end{theorem}
The proof may be seen in Section~\ref{sec:proof}.

The theorem suggests that \(\ell\) should be chosen so as to  maximize
\begin{equation*}
E = \ell^2 A(\ell) = 2 \ell^2 \Phi\left(-\frac{\ell^{\nu}\sqrt{K\Sigma}}{2}\right).
\end{equation*}
In practice \(\Sigma\) will be unavailable; however it is well known and remarkable that, in situations like this, the optimization problem for \(E\) may be solved by using \(A\) as an independent variable in lieu of \(\ell\). In terms of \(A\),
\begin{eqnarray}\label{eq:EA}
E =  \frac{2^{2/\nu}}{(K\Sigma)^{1/\nu}}A \left(\Phi^{-1}(1-\frac{A}{2})\right)^{2/\nu} ,
\end{eqnarray}
and clearly the value \(A_{\rm opt}\) that maximizes \(E\), while changing with \(\nu\), is independent of \(\Sigma\) and \(K\) and therefore \emph{of the target and the scaling constants}. Values of \(A_{\rm opt}\) are provided in Table~\ref{table}. Recall that necessarily reversible integrators cannot have an even value of \(\nu\).

\begin{table}[ht]
\centering
\begin{tabular}{lc}
\toprule
$\nu$ & $A_{\rm opt}$ \\
\midrule
1 & 0.234 \\
3 & 0.574 \\
5 & 0.704 \\
7 & 0.773 \\
9 & 0.816\\
\bottomrule
\end{tabular}
\caption{Optimal acceptance probability as a function of $\nu$.}
\label{table}
\end{table}

With implicit algorithms, like MAIMLA, enlarging the value of the  step length \(h\) may result in the need for more iterations to solve the nonlinear equations at each step. In that scenario, maximizing the squared jumping distance per step as above will be different from maximizing the squared jumping distance per unit of computational work. It is likely that in practice for implicit algorithms the acceptance rate should be chosen to be slightly higher than the values we just quoted.
\begin{remark}
The samplers we have been discussing, although presented in a Hamiltonian framework, are not to be confused with Hamiltonian Monte Carlo (HMC) algorithms. These also integrate the Hamiltonian dynamics with a volume preserving, reversible discretisation, but to generate a proposal, rather than taking a single time step with \(\psi_h\), they integrate over a fixed time interval \(0\leq t \leq T\). Therefore the  error relevant to accept/reject is the error after many {\color{black} time} steps, i.e.\ the global error. A result like Theorem~\ref{th:acta} holds, but the value of \(\nu\) is \emph{even} (Verlet has \(\nu = 2\)). The scaling is \(h_d = \ell/d^{1/(2\nu)}\) in the i.i.d. case. The \lq\lq mean equals half the variance property holds\rq\rq\ which explains why the acceptance probability has an expression of the form \eqref{eq:main}. However in HMC the quantity to be maximized for optimal scaling is \(\ell A\) rather than \(\ell^2A\) (the progress of the chain is determined by the duration \(T\) and the number of time steps to get a proposal behaves like \(\ell^{-1}\)).
\end{remark}
\subsection{Preconditioned dynamics}\label{sec:preconditioning}

To sample from the target \eqref{eq:target}, one may consider, rather than \eqref{eq:extendedtarget}, the extended target
\begin{equation}\label{eq:extendedtargetmass}
\Pi^M(q,p)\propto \exp(-H^M(q,p)),\qquad H^M(q,p) = \frac{1}{2}p^TM^{-1} p+V(q),\qquad (q,p)\in\bbR^m\times\bbR^m,
\end{equation}
where \(M\) is a positive definite symmetric \(m\times m\) matrix, referred to as mass matrix. The auxiliary momentum variable \(p\) now has distribution \(N(0,M)\) and Hamilton's equations  read
\begin{equation}\label{eq:hamiltonianequationsM}
\frac{d}{dt} q = M^{-1}p,\qquad\qquad \frac{d}{dt}p = -\nabla V(q).
\end{equation}
The integrators described in Section~\ref{sec:framework} are easily extended to this more general setting and, after defining the acceptance probability
\begin{equation}\label{eq:acceptanceprobabilityM}
a(q,p,h) = \min(1,\exp(-\Delta^M(q,p,h))), \qquad \Delta^M(q,p,h) = H^M(\textcolor{blue}{q^\star},p^\star)-H^M(q,p),
\end{equation}
and considering the marginal on \(q\), they
generate Markov chains  reversible with respect to the target \eqref{eq:target}.

For reasons of brevity, we only write down the formulas  for the \typo{St\"{o}rmer--Verlet} integrator. This  reads:
\begin{equation*}
p_{1/2} = p -\frac{h}{2}\nabla V(q), \qquad q^\star = q+hM^{-1}p_{1/2},\qquad p^\star = p_{1/2} -\frac{h}{2}\nabla V(q^\star).
\end{equation*}
Elimination of \(p_{1/2}\) and \eqref{eq:hdelta} yield
\begin{equation*}
q^\star = q - \frac{\delta}{2} M^{-1}\nabla V(q)+ \sqrt{\delta} M^{-1}p, \qquad p \sim N(0,M),
\end{equation*}
or
\begin{equation*}
q^\star = q - \frac{\delta}{2} M^{-1}\nabla V(q)+ \sqrt{\delta} M^{-1/2}z, \qquad \rev{z} \sim N(0,I_m),
\end{equation*}
a consistent discretisation of the \emph{preconditioned} Langevin equation
\begin{equation*}
dq(t) = -\frac{1}{2}M^{-1}\nabla V(q(t))\,dt+ M^{-1/2}\, dW(t),
\end{equation*}
which preserves the target \eqref{eq:target}. Thus the introduction of mass matrices different from \(I_m\) in the Hamiltonian formulation provides preconditioning for the dynamics.

Let us study the effect of scaling the potential as in Section~\ref{sec:scaling}. Assume that to sample from \(V(q/\lambda)\), we choose the mass matrix \(M\) to be \(\lambda^{-2}I_m\). From a physical point of view, this means increasing the mass  as \(\lambda\) decreases and the potential becomes steeper; this avoids large values of the velocity \((d/dt) q\), see  \cite{BPSSS11, CSSS22} for further discussion of this point. For this choice of \(M\) the differential equations \eqref{eq:hamiltonianequationsM} are
\begin{equation*}
\frac{d}{dt} q = \lambda^2 p,\qquad\qquad \frac{d}{dt}p = -\frac{1}{\lambda}\nabla V(q/\lambda),
\end{equation*}
and the scaling of variables
\begin{equation*}
\bar{q} = q/\lambda,\qquad \bar{p} = \lambda p
\end{equation*}
transforms them into
\begin{equation*}
\frac{d}{dt} \bar{q} = \bar{p},\qquad\qquad \frac{d}{dt}\bar{p} = - \nabla V(\bar{q}),
\end{equation*}
i.e.\ the Hamiltonian differential equations for the unscaled potential when the mass matrix is chosen to be the identity. Therefore, for an equivariant integrator and \(M = \lambda^{-2}I_m\):
\begin{equation*}
\Delta^M(q,p,h) = \Delta(q/\lambda,\lambda p, h),
\end{equation*}
where \(\Delta\) corresponds to the integrator as applied to \eqref{eq:extendedtarget}. In addition we have the following analogue of Lemma \ref{lem:changein expectation}, whose elementary proof will not be given:

\begin{lemma} Let \(F(q,p)\) be a function defined in \(\bbR^m\times\bbR^m\) and \(\lambda>0\) and define \(M = \lambda^{-2}I_m\). The expectation of
\(F(q/\lambda,\lambda p)\) when
\((q,p)\) is distributed according to the density \(\propto \exp(-H^M(q,p))\) coincides with the expectation of \(F(q,p)\) when
\((q,p)\) has density \(\propto \exp(-H(q,p))\).
\end{lemma}

Therefore, if to sample from the product target \eqref{eq:product target} in \(\bbR^N\), we use the Hamiltonian
\begin{equation}\label{eq:HM}
{\cal H}^M(Q,P) = {\cal H}^M(q_1,\dots,q_d,p_1\dots,p_d) = \frac{1}{2} \sum_{i=1}^d  \|\lambda_ip_i\|^2+\sum_{i=1}^d V(q_i/\lambda_i),
\end{equation}
then the energy error to accept/reject will satisfy
\begin{equation*}
\Delta_d^M(Q,P,h) = \sum_{i=1}^d \Delta^M(q_i,p_i,h) = \sum_{i=1}^d \Delta(q_i/\lambda,\lambda p_i,h).
\end{equation*}
This is very similar to \eqref{eq:Deltad} and reduces the study of \(\Delta_d^M\) to the study of the energy errors for the potential \(V\) and unit mass matrix.
However, we note that here and as distinct from \eqref{eq:Deltad}, a common value of \(h\) features in the different \(\Delta(q_i/\lambda,\lambda p_i,h)\) being summed: preconditioning ensures that the different \(q_i\) evolve in the same time scale.
Therefore the present situation is analogous to the setting where in Theorem~\ref{th:main} all the \(\lambda_i\) take the value 1. By mimicking the proof of Theorem~\ref{th:main}, one derives the following result:

\begin{theorem}\label{th:mainM}Consider the MCMC algorithm to sample from the product \eqref{eq:product target} based on an equivariant, volume preserving and reversible integrator for the Hamiltonian equations for the Hamiltonian function \eqref{eq:HM}. Assume that the integrator, when applied with unit mass matrix to the potential \(V\) in \(\bbR^m\) satisfies Conditions 1 and 2.
Choose \(h_d = \ell/d^{1/(2\nu)}\), for some constant \(\ell>0\).
Then, at stationarity, the expectation  of the acceptance probability  converges to the limit
\begin{equation*}
A = 2 \Phi\left(-\frac{\ell^{\nu}\sqrt{\Sigma}}{2}\right),
\end{equation*}
where  \(\Sigma\) is defined in Theorem~\ref{th:acta}.
\end{theorem}

Note that there are \emph{no hypotheses on the} \(\lambda_i\) and that, when preconditioning is used, the scaling for \(h_d\) coincides with the one found in the nonpreconditioned case when all the \(q_i\) share a common distribution. In the particular case of the MALA algorithm with \(m=1\) and \(\lambda_i = 1/i^\kappa\) the formula for \(h_d\) was presented in \cite[Theorem~5.2]{BS09}, and the value of  \(A\) was given in \cite{BRS09}.

Let us now discuss the squared jumping distance for preconditioned algorithms. With preconditioning, due to the presence of the mass matrix, \(q_i^\star-qi \approx h_d \lambda_i^2 p_i\) and \(p_i\sim N(0,(1/\lambda_i^2)I_m)\). Therefore, we would expect that
\begin{equation*}
\bbE(\|q_i^\star-q_i\|^2) \approx m h_d^2 \lambda_i^2;
\end{equation*}
in this way
the mean quadratic displacement from current state to proposal varies with \(i\) and is proportional to \(\lambda_i^2\), and therefore to the second moment of \(q_i\), which has distribution \(\exp(-V(q_i/\lambda_i))\). (Achieving this proportionality is precisely the motivation for preconditioning.)
The last display should be compared with the non-preconditioned case where (cf.\ Condition 3) regardless of the value of \(i\)
\begin{equation*}
\bbE(\|q_i^\star-q_i\|^2) \approx m h_d^2.
\end{equation*}
Due to this difference, the metric defined in \eqref{eq:sjd}, based on an average with respect to \(i\), makes little sense for preconditioned dynamics. One may instead use, for \(i\leq d\), the alternative metric
\begin{equation*}
{\cal SJD}_d^i = \bbE(\|q_i^\star-q_i\|^2 \bbI_{\{a_d(Q,P,h_d) \leq U\}}),\qquad U\sim{\rm Uniform}(0,1).
\end{equation*}

With the scaling \(h_d= \ell/d^{1/(2\nu)}\) from Theorem~\ref{th:mainM}, one has
\begin{eqnarray*}
{\cal SJD}_d^i \sim m \lambda_i^2 d^{-1/\nu}\ell^2 A, \quad d \rightarrow \infty.
\end{eqnarray*}
The main difficulty in proving this estimate stems from  \(\|q_i^\star-q_i\|^2\) and \(\bbI_{\{\Delta_d(Q,P,h_d) \leq U\}}\)  \emph{not} being independent. This obstacle may be circumvented by using the technique in the proof of \cite[Proposition~3.8]{BPRSSS13}, that takes into account that, for \(d\) large, \(\bbI_{\{\Delta_d(Q,P,h_d) \leq U\}}\)  is \lq almost\rq\ independent of the single component \((q_i,p_i)\). Details will not be given.
The  proof used in this paper for Theorem~\ref{th:sjd} does not have to address the dependence issue because it deals with an average  \((1/d)\sum_i\|q_i^\star-q_i\|^2\approx (1/d)h_d^2\sum_i\|p_i\|^2\) whose limit, being constant by the weak law of large numbers, is stochastically independent from the indicator function.

  In any case, as in the non-preconditioned case, the best value of \(\ell\) is determined by maximising \(\ell^2A\). As discussed above, the maximum is achieved for the values reported in Table~\ref{table}.

\section{General Metropolis--Hastings algorithms}\label{sec:optimal_MH}

{\color{black}
The preceding material and the construction of MCMC algorithms with enhanced scaling properties to be discussed in the next section depend on using results from the Hamiltonian formalism and the integration of Hamiltonian systems. In particular, the property \lq\lq mean equals half the variance\rq\rq\ in Theorem~\ref{th:acta}, which is the foundation of the main scaling result in Theorem~\ref{th:main}, has been proved using the preservation of volume and reversibility properties of the integrator. It is remarkable that, as shown in \cite{VK21} (see also \cite{VLZ23}) the
\lq\lq mean equals half the variance\rq\rq\ property may be proved as a direct consequence of the Metropolis-Hastings formula and therefore applies to \lq\lq general\rq\rq\ MCMC algorithms, i.e.\ to cases that cannot be treated within the Hamiltonian framework.
In this section we comment on how to extend the material in the preceding sections to a non-Hamiltonian setting. For reasons of brevity  we limit ourselves to a terse summary and omit a number of technical details.
}

For the target \eqref{eq:target}, consider a proposal of $q^\star$ from a position $q$ according to the probability density function $\wp_h(q,q^\star)$. Here $h$ is a small parameter for the algorithm, perhaps related to the \rev{standard deviation} of the proposal. Acceptance is based on the ratio
\begin{equation*}
     \frac{\exp(-V(q^\star))\: \wp_h(q^\star,q)}{\exp(-V(q))\:\wp_h(q,q^\star)};
\end{equation*}
swapping \(q\) and \(q^\star\) takes this ratio into its inverse.
To parallel the notation of the preceding sections, we introduce
\begin{equation*}
    \Delta(q,q^\star,h) = V(q^\star)-V(q)-\log\left(\frac{\wp_h(q^\star,q)}{\wp_h(q,q^\star)}\right)
\end{equation*}
and then the acceptance probability is $\min\{1,\exp(-\Delta(q,q^\star,h))\}$. Note that
\begin{equation}\label{eq:oppositeqstar}
\Delta(q,q^\star,h) = -\Delta(q^\star,q,h);
\end{equation}
 this  just rephrases the symmetry of the Metropolis-Hastings formula, but  is the foundation of all subsequent developments. We encountered a similar relation in  \eqref{eq:opposite} when working in the Hamiltonian framework.

The relation \eqref{eq:oppositeqstar}  implies that when computing \(\bbE(\Delta)\) much cancellation takes place. More precisely, after defining
\begin{equation*}
    \mu(h) = \mathbb{E}[\Delta(q,q^\star,h)], \qquad s(h) = \mathbb{E}[\Delta(q,q^\star,h)^2], \qquad \sigma^2(h)=s(h)-\mu(h)^2,
\end{equation*}
(expectations $\mathbb{E}$ are with respect to the probability measure \(\Pi(q,q^\star)\propto\exp(-V(q))\wp(q,q^\star)dq dq^\star\))
we have the following
 analogue to Theorem~\ref{th:acta} {\color{black}(cf.\ the similar developments in \cite[Section~3]{VK21}). The proof is given in Section~\ref{sec:proof}.}
\begin{theorem}\label{thm:mean_variance_relation}
We will make the following assumption: for some \(\nu>0\),
 \(\Sigma = \lim_{h\to 0} s(h)/{h^{2\nu}}\) exists is finite and nonzero. Moreover, there exists a real function $D(q,q^\star)$ such that
\begin{equation*}
    \sup_{0< h\leq 1} \frac{|\Delta(q,q^\star,h)|^3}{h^{3\nu}} \leq D(q,q^\star)
\end{equation*}
and $D$ is integrable with respect to the probability measure \(\Pi(q,q^\star)\).
 Then
\begin{equation*}
    \lim_{h\to 0} \frac{\mu(h)}{h^{2\nu}} =\frac{\Sigma}{2}, \qquad \lim_{h\to 0} \frac{\sigma^2(h)}{h^{2\nu}} = \Sigma.
\end{equation*}
\end{theorem}

{\color{black} This result provides a basis for the study of
 the scaling for products as in \eqref{eq:product target}. To simplify the exposition, we only consider the case where the target is given as an i.i.d.\ product, with $\lambda_i=1$ for all $i$.}  (The scaled case may be dealt with via dimensionality arguments as we did in the Hamiltonian framework.) We only give the main idea. The acceptance probability may  be written as
\begin{equation*}
    a_d(Q,Q^\star,h) = \min\{1,e^{-\Delta_d(Q,Q^\star,h)}\}, \qquad \Delta_d = \sum_{i=1}^d \Delta(q_i,q_i^\star,h).
\end{equation*}
As a consequence of Theorem~\ref{thm:mean_variance_relation}, $\{\Delta(q_i,q_i^\star,h)\}_{i=1}^d$ is a sequence of i.i.d. \(d\) random variables with mean approximately $(1/2) h^{2\nu}\Sigma$ and variance approximately $ h^{2\nu}\Sigma$. Setting $h_d = \ell d^{-1/(2\nu)}$ and applying the Central Limit Theorem, the distributions of \(\Delta_d\) will approach the distribution of a normal variable \(\Delta_\infty\) of mean \((1/2)\ell^{2\nu}\Sigma\) and variance \(\ell^{2\nu}\Sigma\) and, by arguing as in Theorem~\ref{th:main}, the limit acceptance probability will have the value
\[
A = 2\Phi\left(-\frac{\ell^\nu\sqrt{\Sigma}}{2}\right).
\]
Of course this is exactly \eqref{eq:main} (for the case at hand with \(K=1\)). We note, {\color{black} once more,} that the expression for \(A\) is a consequence of the \lq\lq mean  equal half the variance\rq\rq\ property of the limit normal variable \(\Delta_\infty\).

{\color{black} For algorithms where, in analogy with Condition 3 above, \(\bbE((q^\star-q)^2)\) is, to leading order, a constant multiple of \(h^2\), the optimal acceptance rate is obtained by maximizing \(\ell^2A\) and
therefore  the values of the optimal acceptance probability reported in Table~\ref{table} are valid in the present scenario.}

As a first, simple  illustration, we apply these considerations to RWM, a sampler we already discussed within the Hamiltonian framework. This has
\begin{equation*}
    \Delta(q,q^\star,h) = V(q^\star)-V(q).
\end{equation*}
Conditional on \(q\), the proposal has a distribution \(N(q, h^2I_m)\) and then, with \(p\sim N(0,I_m)\),
\begin{equation*}
    s(h) = \mathbb{E}[\Delta(q,q^\star)^2] = \mathbb{E}[(V(q+hp)-V(q))^2].
\end{equation*}
Therefore, for suitably smooth \(V\), \(s(h)/h^{2} \) approaches \(\Sigma = \bbE\big(\sum_i (\partial V/\partial q^ip^i)^2\big)\) as \(h\rightarrow 0\); so \(\nu=1\) in Theorem~\ref{thm:mean_variance_relation} and we then recover the RWM results in Section~\ref{sec:main}.

As a second illustration, we consider fMALA, a modification of MALA constructed in \cite{DRVZ17} to achieve  \(\delta_d = \ell^2/d^{1/5}\) scaling in the i.i.d. product case. This scaling for fMALA (and for its  variants) is attained at the price of  using, in addition to the gradient, the Hessian and the tensor of third derivatives of the potential. We  have not been able to  cast fMALA in the Hamiltonian framework, but the material in this section may be applied. By expanding \(s(h)\), one shows that the algorithm satisfies the requirements in Theorem~\ref{thm:mean_variance_relation} with \(\nu=5\); in fact fMALA was constructed by annihilating the lower order terms in the expansion of \(s(h)\) in powers of \(h\). The  material in this section  then implies that the variance of the proposal has to be scaled as \(\delta_d = \ell^2/d^{1/5}\) and the
optimal acceptance rate is 0.704. It is perhaps of interest to observe that \cite{DRVZ17} implicitly proves the \lq\lq mean  equal half the variance\rq\rq\ property (see formulas (47)--(48) in that paper), but does so  by finding  the expansion of \(\bbE(\Delta)\) in powers of \(h\) to order 10, a task that required the use of a symbolic algebra package. Also \cite{DRVZ17}  was limited to the univariate case \(m=1\).

\section{Missing proofs}\label{sec:proof}

{\color{black}\emph{Proof of Theorem~\ref{th:acta}:}
   The proof is based on the observation (see \cite[Proposition~6.1]{BRSS18}) that, for a volume preserving and reversible integrator,
   \begin{equation}\label{eq:new1}
   \int_{\mathbb{R}^m} \int_{\mathbb{R}^m}\Delta(q,p,h) e^{-H(q,p)}dqdp = -\int_{\mathbb{R}^m} \int_{\mathbb{R}^m}\Delta(q,p,h)e^{-\Delta(q,p,h)}e^{-H(q,p)}dqdp
   \end{equation}
   while, for any even real function \(\phi\) of a real variable,
   \begin{equation}\label{eq:new2}
   \int_{\mathbb{R}^m} \int_{\mathbb{R}^m}\phi(\Delta(q,p,h)) e^{-H(q,p)}dqdp = \int_{\mathbb{R}^m} \int_{\mathbb{R}^m}\phi(\Delta(q,p,h))e^{-\Delta(q,p,h)}e^{-H(q,p)}dqdp.
   \end{equation}

   From \eqref{eq:new1},
        \begin{equation*}
        \mu(h)= \frac{1}{2}\int_{\mathbb{R}^m} \int_{\mathbb{R}^m} \Delta(q,p,h) \left(1-e^{-\Delta(q,p,h)}\right)  e^{-H(q,p)} dq dp,
    \end{equation*}
    where we note that, for small \(\Delta\), the integrand is of the order of \(\Delta^2\) since \(1-\exp(-\Delta)\approx \Delta\). To make this observation more precise, we write
    \begin{align*}
        &\frac{1}{h^{2\nu}}\left|\mu(h) - \frac{1}{2}s(h)\right| \leq\\&\qquad\qquad \frac{1}{2}\int_{\mathbb{R}^m} \int_{\mathbb{R}^m} \left|\frac{1}{h^{2\nu}} \Delta(q,p,h) \left(1-e^{-\Delta(q,p,h)}-\Delta(q,p,h)\right) \right| e^{-H(q,p))} dq dp,
    \end{align*}
    use the bound
    \begin{equation}\label{eq:exp_bound}
        |u||e^{u}-1-u| \leq \frac{1}{2}|u|^3(e^u+1), \qquad u\in \mathbb{R},
    \end{equation}
    and proceed as follows:
\begin{align*}
        \frac{1}{h^{2\nu}}\left|\mu(h) - \frac{1}{2}s(h)\right| &\leq \frac{1}{4h^{2\nu}}\int_{\mathbb{R}^m} \int_{\mathbb{R}^m} \left| \Delta(q,p,h)\right|^3
        (e^{-\Delta(q,p,h)}+1)  e^{-H(q,p)} dq dp\\
        & = \frac{1}{4h^{2\nu}}\int_{\mathbb{R}^m} \int_{\mathbb{R}^m} \left| \Delta(q,p,h)\right|^3
        e^{-\Delta(q,p,h)}  e^{-H(q,p)} dq dp\\
        &\qquad\qquad +\frac{1}{4h^{2\nu}}\int_{\mathbb{R}^m} \int_{\mathbb{R}^m} \left| \Delta(q,p,h)\right|^3
        e^{-H(q,p)} dq dp\\
        & = \frac{1}{2h^{2\nu}}\int_{\mathbb{R}^m} \int_{\mathbb{R}^m} \left| \Delta(q,p,h)\right|^3
        e^{-\Delta(q,p,h)}  e^{-H(q,p)} dq dp
    \end{align*}
    The last equality is a consequence of \eqref{eq:new2}.
    Therefore
    \begin{align*}
        \frac{1}{h^{2\nu}}\left|\mu(h) - \frac{1}{2}s(h)\right| \leq & \frac{h^\nu}{2}\int_{\mathbb{R}^m} \int_{\mathbb{R}^m} \frac{\left| \Delta(q,p,h)\right|^3}{h^{3\nu}} e^{-H(q,p)}  dq dp.
    \end{align*}
    By Condition 2 the integral is finite and bounded in $h$, therefore letting $h\to 0$,
     \begin{equation}
        \lim_{h\to 0}\frac{1}{h^{2\nu}}\left|\mu(h) - \frac{1}{2}s(h)\right| =0,
     \end{equation}
    and the result follows easily.
} 
\bigskip

\emph{Proof of Theorem~\ref{th:main}:}

We begin by proving that for the triangular array in \eqref{eq:row}, the variance \({\rm Var}_d \) of the sum \(\Delta_d\) of the independent variables in the \(d\)-th row converges to \(\ell^{2\nu}K\Sigma\) as \(d\rightarrow\infty\). We decompose as
\begin{equation*}
{\rm Var}_d = \sum_{i=1}^d \sigma^2(h_d/\lambda_i) =
\sum_{i=1}^d \left(\sigma^2(h_d/\lambda_i)- \frac{h_d^{2\nu}}{\lambda_i^{2\nu}} \Sigma\right)+
\sum_{i=1}^d \frac{h_d^{2\nu}}{\lambda_i^{2\nu}} \Sigma.
\end{equation*}
and recall from the discussion following Theorem~\ref{th:main}, that the last of these sums converges to \(\ell^{2\nu}K\Sigma\). On the other hand
\begin{eqnarray}
\sum_{i=1}^d \left(\sigma^2(h_d/\lambda_i)- \frac{h_d^{2\nu}}{\lambda_i^{2\nu}} \Sigma\right)& = &
h_d^{2\nu} \sum_{i=1}^d\frac{1}{\lambda_i^{2\nu}} \left(\frac{\sigma^2(h_d/\lambda_i)}
{ \frac{h_d^{2\nu}}{\lambda_i^{2\nu}} }-\Sigma\right)\nonumber\\
&=& \ell^{2\nu}d^{-\gamma}\sum_{i=1}^d\frac{1}{\lambda_i^{2\nu}} \left(\frac{\sigma^2(h_d/\lambda_i)}
{ \frac{h_d^{2\nu}}{\lambda_i^{2\nu}} }-\Sigma\right)\nonumber\\
&\leq&\ell^{2\nu}\left( d^{-\gamma} \sum_{i=1}^d\frac{1}{\lambda_i^{2\nu}}\right) \max_{i=1,\dots,d}\left|\frac{\sigma^2(h_d/\lambda_i)}
 {\frac{h_d^{2\nu}}{\lambda_i^{2\nu}}} -\Sigma\right|.\label{eq:boundvariance}
\end{eqnarray}
Now, according to \eqref{eq:conditionmax}, the quotients \(h_d/\lambda_i\), \(i=1,\dots,d\) converge uniformly to zero:
\begin{equation}\label{eq:uniformity}
\lim_{d\rightarrow \infty}\max_{1,\dots,d}\frac{h_d}{\lambda_i} = \lim_{d\rightarrow \infty} \max_{1,\dots,d} \ell \left(d^{-\gamma}\frac{1}{\lambda_i^{2\nu}}\right)^{1/(2\nu)} = 0.
\end{equation}
Therefore, from Theorem~\ref{th:acta},
\begin{equation*}
\lim_{d\rightarrow\infty}\max_{i=1,\dots,d}\left|\frac{\sigma^2(h_d/\lambda_i)}
 {\frac{h_d^{2\nu}}{\lambda_i^{2\nu}}} -\Sigma\right| = 0,
\end{equation*}
which, in tandem with \eqref{eq:gamma}, implies that \eqref{eq:boundvariance} converges to zero. We conclude that \({\rm Var}_d\rightarrow \ell^{2\nu}K\Sigma\), as we had announced. In a similar way, one proves that the expectation of \(\Delta_d\) converges to \(\ell^{2\nu}K\Sigma/2\).

It remains to show that the Lindeberg condition is fulfilled. Fix \(\epsilon>0\) and consider the centered variables \(\Delta(q_i/\lambda_i,p_i,h_d/\lambda_i)-\mu(h_d/\lambda_i\)), \(i=1,\dots,d\), and the quotient
\begin{equation*}
\frac
{\sum_{i=1}^d\bbE \Big(
\big(\Delta(q_i/\lambda_i,p_i,h_d/\lambda_i)-\mu(h_d/\lambda_i)\big)^2
\bbI_{\{|\Delta(q_i/\lambda_i,p_i,h_d/\lambda_i)-\mu(h_d/\lambda_i)| > \epsilon{\rm Var}_d^{\frac{1}{2}}\}}
\Big)}
{{\rm Var}_d}.
\end{equation*}
As we have just proved, the denominator has a finite limit and therefore our task is to prove that numerator converges to zero.
By using Lemma~\ref{lem:changein expectation}, this numerator may be written as
\begin{equation*}
\sum_{i=1}^d\bbE \Big(
\big(\Delta(q,p,h_d/\lambda_i)-\mu(h_d/\lambda_i)\big)^2
\bbI_{\{|\Delta(q,p,h_d/\lambda_i)-\mu(h_d/\lambda_i)| > \epsilon{\rm Var}_d^{\frac{1}{2}}\}}
\Big),
\end{equation*}
where now \((q,p\)) has density \(\propto\exp(-H(q,p))\), or, after rearrangement

\begin{equation*}
\sum_{i=1}^d
\sigma^2(h_d/\lambda_i)
\bbE \Big(
\Big(\frac{\Delta(q,p,h_d/\lambda_i)-\mu(h_d/\lambda_i)}{\sigma(h_d/\lambda_i)}\Big)^2
\bbI_{\{|\Delta(q,p,h_d/\lambda_i)-\mu(h_d/\lambda_i)| > \epsilon{\rm Var}_d^{\frac{1}{2}}\}}
\Big).
\end{equation*}
Since we know that \(\sum_i \sigma^2(h_d/\lambda_i)\) has a finite limit, the proof will be complete if we show that
\begin{equation*}
\lim_{d\rightarrow \infty} \max_{i=1,\dots,d} \bbE \Big(
\Big(\frac{\Delta(q,p,h_d/\lambda_i)-\mu(h_d/\lambda_i)}{\sigma(h_d/\lambda_i)}\Big)^2
\bbI_{\{|\Delta(q,p,h_d/\lambda_i)-\mu(h_d/\lambda_i)| > \epsilon{\rm Var}_d^{\frac{1}{2}}\}}
\Big)= 0
\end{equation*}
or, recalling the uniform convergence in \eqref{eq:uniformity}, if we prove that
\begin{equation*}
\lim_{t\rightarrow 0} \bbE \Big(
\Big(\frac{\Delta(q,p,t)-\mu(t)}{\sigma(t)}\Big)^2
\bbI_{\{|\Delta(q,p,t)-\mu(t)| > \epsilon{\rm Var}_d^{\frac{1}{2}}\}}
\Big)= 0.
\end{equation*}
In order to do so, apply the dominated convergence theorem, after noting that
from Condition 1 and Theorem~\ref{th:acta},
\begin{equation*}
\lim_{t\rightarrow 0} \frac{\Delta(q,p,t)-\mu(t)}{\sigma(t)} = \frac{\alpha(q,p)}{\sqrt{\Sigma}},
\end{equation*}
and
\begin{equation*}
\lim_{t\rightarrow 0} \bbI_{\{|\Delta(q,p,t)-\mu(t)| > \epsilon{\rm Var}_d^{\frac{1}{2}}\}}  =
\bbI_{\{\frac{|\Delta(q,p,t)-\mu(t)|}{\sigma(t)} > \epsilon\frac{{\rm Var}_d^{\frac{1}{2}}\}}{\sigma(t)}} = 0,
\end{equation*}
because \(\sigma(t)\rightarrow 0\). To dominate the integrand, observe that  \(\Delta^2/\sigma^2\) behaves like \((\Delta/t^\nu)^2\) and is therefore upper bounded by the integrable function \(D^{2/3}\), where \(D\) is from Condition 2.\bigskip

\emph{Proof of Theorem~\ref{th:sjd}:}

By integrating with respect to \(U\):
\begin{equation*}
{\cal SJD}_d = \frac{1}{d} \sum_{i=1}^d \bbE(\|q_i^\star-q_i\|^2 a_d(Q,P,h_d)).
\end{equation*}
Then, we have
\begin{equation*}
d^{\gamma/\nu}{\cal SJD}_d = \ell^2\frac{1}{d} \sum_{i=1}^d \bbE\left(\|\frac{1}{h}(q_i^\star-q_i)\|^2 a_d(Q,P,h_d)\right),
\end{equation*}
an expression that may be decomposed as the sum of
\rev{\begin{equation}\label{eq:aux1}
\ell^2\frac{1}{d} \sum_{i=1}^d \bbE(\|p_i\|^2a_d(Q,P,h_d) ) = \ell^2 \bbE \Big(\frac{1}{d}\sum_{i=1}^d \|p_i\|^2a_d(Q,P,h_d)\Big)
\end{equation}}
and a \rev{remainder}
\begin{equation*}
\ell^2\frac{1}{d} \sum_{i=1}^d \bbE\left(\Big(\|\frac{1}{h}(q_i^\star-q_i)\|^2-\|p_i\|^2\Big) a_d(Q,P,h_d)\right).
\end{equation*}
In \eqref{eq:aux1}, we know from the proof of \rev{Theorem~\ref{th:main}} and the continuous mapping theorem that
\(a_d(Q,P,h)\) \( = \min(1,\exp(-\Delta_d))\) converges in distribution to \(\min(1,\rev{\exp(-\Delta_\infty)})\). The random variable \((1/d)\sum_i \|p_i\|^2\) has a Gamma distribution \rev{with} parameters \(\alpha= dm/2\), \(\theta = 2/d\). It has mean \(m\) and variance \(2m/d\) and therefore \rev{converges} to the constant \(m\) (as one may alternatively have proved by the weak law of large numbers). By Slutsky's theorem the variables \((1/d)\sum_i \|p_i\|^2a_d\) have a distributional limit and, since their second moments are bounded,
we conclude that \eqref{eq:aux1} converges to the limit
\begin{equation*}
m \ell^2 \lim_{d\rightarrow \infty} \bbE(a_d(Q,P,h_d)) = m\ell^2 A = 2 m \ell^2 \Phi\left(-\frac{\ell^{\nu}\sqrt{K\Sigma}}{2}\right).
\end{equation*}

The proof concludes by checking that the \rev{remainder} converges to zero. In fact, from equivariance and Condition 3,
\begin{equation*}
\bbE\left(\|\frac{1}{h}(q_i^\star-q_i)\|^2-\|p_i\|^2 \right) = r(h_d/\lambda_i)
\end{equation*}
and it is therefore enough to recall the uniform convergence in \eqref{eq:uniformity}.\bigskip

{\color{black}
\emph{Proof of Theorem~\ref{thm:mean_variance_relation}:} This is to be compared with the proof of Theorem~\ref{th:acta}.
    By relabelling the variables, we can write
    \begin{align*}
        \mathbb{E}[\Delta(q,q^\star,h)]&= \int_{\mathbb{R}^m} \int_{\mathbb{R}^m} \Delta(q,q^\star,h) \wp_h(q,q^\star) e^{-V(q)} dq^\star dq \\&= \int_{\mathbb{R}^m} \int_{\mathbb{R}^m} \Delta(q^\star ,q,h) \wp_h(q^\star ,q) e^{-V(q^\star )} dqdq^\star .
    \end{align*}
    Hence
    \begin{equation*}
        \mu(h)= \frac{1}{2}\int_{\mathbb{R}^m} \int_{\mathbb{R}^m} \Delta(q,q^\star ,h) \wp_h(q,q^\star ) e^{-V(q)} dq^\star dq + \frac{1}{2}\int_{\mathbb{R}^m} \int_{\mathbb{R}^m} \Delta(q^\star ,q,h) \wp_h(q^\star ,q) e^{-V(q^\star )} dqdq^\star .
    \end{equation*}
    Now, using the definition of $\Delta$, we can rewrite $V(q^\star )$ in terms of $V(q)$:
    \begin{align*}
        \mu(h)&= \frac{1}{2}\int_{\mathbb{R}^m} \int_{\mathbb{R}^m} \Delta(q,q^\star ,h) \wp_h(q,q^\star ) e^{-V(q)} dq^\star dq \\&\qquad\qquad+ \frac{1}{2}\int_{\mathbb{R}^m} \int_{\mathbb{R}^m} \Delta(q^\star ,q,h) \wp_h(q,q^\star ) e^{-V(q)-\Delta(q,q^\star ,h)} dqdq^\star .
    \end{align*}
    By using the symmetry \eqref{eq:oppositeqstar}, we find
        \begin{equation*}
        \mu(h)= \frac{1}{2}\int_{\mathbb{R}^m} \int_{\mathbb{R}^m} \Delta(q,q^\star ,h) \left(1-e^{-\Delta(q,q^\star ,h)}\right) \wp_h(q,q^\star ) e^{-V(q)} dq^\star dq,
    \end{equation*}
    where we note that, for small \(\Delta\), the integrand is of the order of \(\Delta^2\).  We  now write
    \begin{align*}
        &\frac{1}{h^{2\nu}}\left|\mu(h) - \frac{1}{2}s(h)\right| \leq\\&\qquad\qquad \frac{1}{2}\int_{\mathbb{R}^m} \int_{\mathbb{R}^m} \left|\frac{1}{h^{2\nu}} \Delta(q,q^\star ,h) \left(1-e^{-\Delta(q,q^\star ,h)}-\Delta(q,q^\star ,h)\right) \right| \wp_h(q,q^\star ) e^{-V(q)} dq^\star dq,
    \end{align*}
    use the bound \eqref{eq:exp_bound}
    and continue as follows:
\begin{align*}
        \frac{1}{h^{2\nu}}\left|\mu(h) - \frac{1}{2}s(h)\right| &\leq \frac{1}{4h^{2\nu}}\int_{\mathbb{R}^m} \int_{\mathbb{R}^m} \left| \Delta(q,q^\star ,h)\right|^3
        (e^{-\Delta(q,q^\star ,h)}+1) \wp_h(q,q^\star ) e^{-V(q)} dq^\star dq\\
        & = \frac{1}{4h^{2\nu}}\int_{\mathbb{R}^m} \int_{\mathbb{R}^m} \left| \Delta(q,q^\star ,h)\right|^3
        e^{-\Delta(q,q^\star ,h)} \wp_h(q,q^\star ) e^{-V(q)} dq^\star dq\\
        &\qquad\qquad +\frac{1}{4h^{2\nu}}\int_{\mathbb{R}^m} \int_{\mathbb{R}^m} \left| \Delta(q,q^\star ,h)\right|^3
        \wp_h(q,q^\star ) e^{-V(q)} dq^\star dq\\
        & = \frac{1}{4h^{2\nu}}\int_{\mathbb{R}^m} \int_{\mathbb{R}^m} \left| \Delta(q,q^\star ,h)\right|^3
        e^{-V(q^\star )} \wp_h(q^\star ,q)  dq^\star dq\\
        &\qquad\qquad +\frac{1}{4h^{2\nu}}\int_{\mathbb{R}^m} \int_{\mathbb{R}^m} \left| \Delta(q,q^\star ,h)\right|^3
        \wp_h(q,q^\star ) e^{-V(q)} dq^\star dq.
    \end{align*}
    Swapping \(q\) and \(q^\star\) in the last but one integral and invoking \eqref{eq:oppositeqstar} once more, we can write this as
    \begin{align*}
        \frac{1}{h^{2\nu}}\left|\mu(h) - \frac{1}{2}s(h)\right| \leq & \frac{h^\nu}{2}\int_{\mathbb{R}^m} \int_{\mathbb{R}^m} \frac{\left| \Delta(q,q^\star ,h)\right|^3}{h^{3\nu}} e^{-V(q)} \wp_h(q,q^\star )  dq^\star dq.
    \end{align*}
    By the assumption the integral is finite and bounded in $h$, therefore letting $h\to 0$
     \begin{equation}
        \lim_{h\to 0}\frac{1}{h^{2\nu}}\left|\mu(h) - \frac{1}{2}s(h)\right| =0,
     \end{equation}
    and the result follows.
}

\section{Outlook: high order algorithms for better scaling }\label{sec:outlook}

In Section~\ref{sec:framework} we presented three volume preserving, reversible integrators that generate MCMC algorithms. The associated values of \(\nu\) were 1, 3, 3. We now discuss alternative integrators with higher values of \(\nu\) and therefore better scaling properties.

Many symplectic, reversible integrators for Hamiltonian problems have been suggested in the literature \cite{sanz2018numerical}. Since they are automatically volume preserving, all of them may be used as described in Section~\ref{sec:framework}. The class of reversible \emph{splitting} integrators \cite{BRSS18} is of particular interest. Those algorithms are simple, easy to implement generalizations of the Verlet scheme that use several evaluations of \(\nabla V\) to generate a single proposal. As an illustration \typo{consider} as in \cite{CSS17} the family of splitting integrators \((q^\star,p^\star) = \psi_h(q,p)\) described by the equations
\begin{eqnarray*}
p^{(1)} &=& p- h (1/2-b)\nabla V(q),\\
q^{(1)} &=& q+ h a p^{(1)},\\
p^{(2)} &=& p^{(1)}- h b\nabla V(q^{(1)}),\\
q^{(2)} &=& q^{(1)}+ h (1-2a)p^{(2)},\\
p^{(3)} &=& p^{(2)}- h b\nabla V(q^{(2)}),\\
q^\star &=& q^{(2)}+ h ap^{(3)},\\
p^\star &=& p^{(3)}- h (1/2-b)\nabla V(q^\star),
\end{eqnarray*}
where \(a\) and \(b\) are real parameters. On acceptance, the first evaluation of \(\nabla V\) to be used for the next proposal coincides with the last evaluation \(\nabla V(q^\star)\) in the current step. On rejection, the first evaluation of \(\nabla V\) to generate the next proposal coincides with the first evaluation \(\nabla V(q)\) in the current step. Thus each step of the Markov chain, except the very first, uses three gradient evaluations.
\begin{figure}
\centering
\includegraphics[scale = .6]{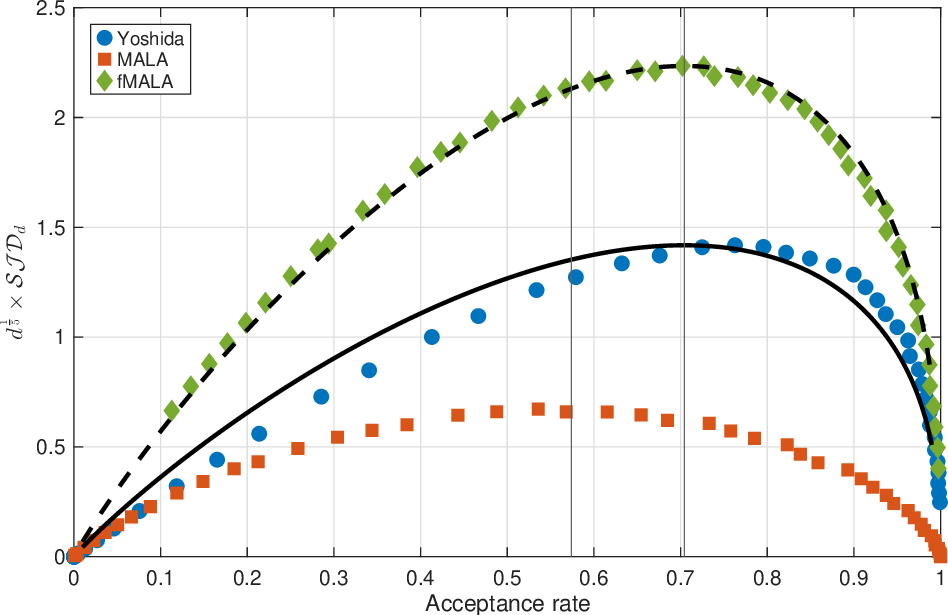}\\
\vspace{0.5 cm}
\includegraphics[scale = .6]{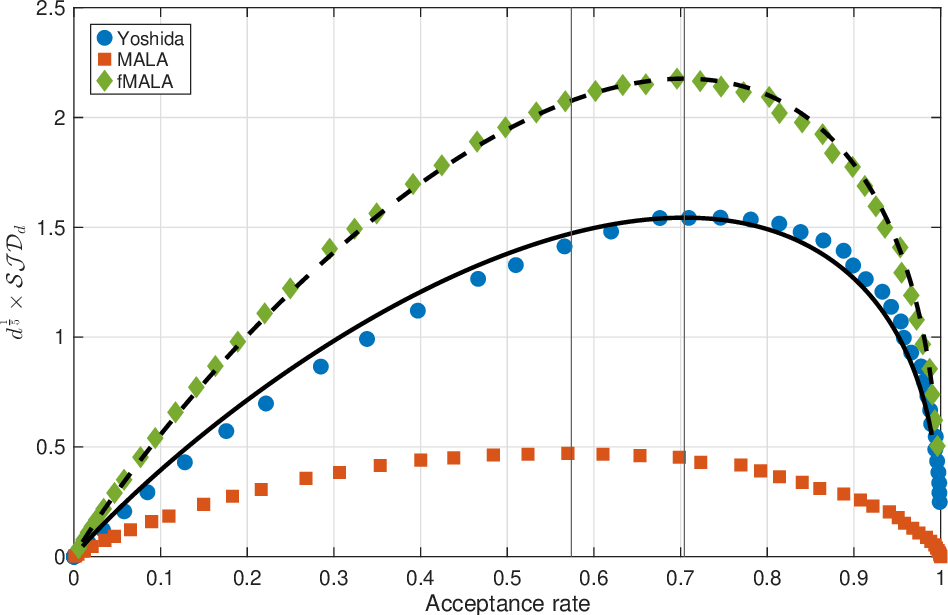}\\
\vspace{0.5 cm}
\includegraphics[scale = .6]{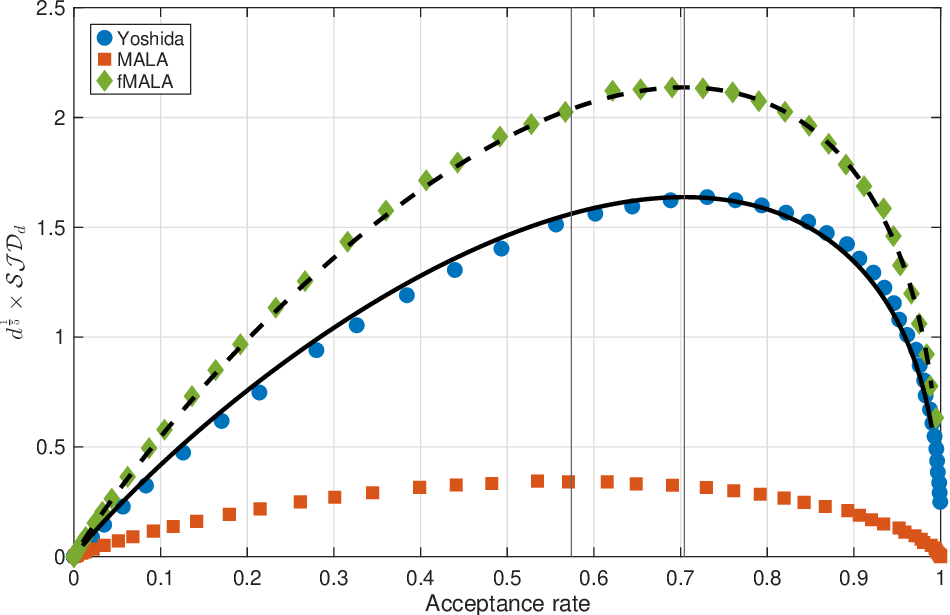}
\caption{Scaled Square Jumping Distance as a function of the acceptance rate, \(d= 10^{3}\) (top), \(d=10^{4}\) (middle), \(d=10^{5}\) (bottom)}
\label{figure}
\end{figure}
For the choice
\begin{equation*}
a = \frac{1}{2-\sqrt[3]{2}},\qquad b = \frac{1-a}2
\end{equation*}
sometimes attributed to Yoshida,
the integrator \rev{is} fourth order accurate with \(\nu= 5\). For the i.i.d.\ case with \(\lambda_i=1\), the scaling is \(h_d = \ell/ d^{1/10}\) or
\(\delta_d = \ell^2/d^{1/5}\), a marked improvement on \(\delta_d = \ell^2/d\) for RWM and \(\delta_d = \ell^2/d^{1/3}\) for MALA. By writing \(q^\star\) as a function of \(q\) and \(p\) and assuming \(p\sim N(0,I_m)\), it is easy to check that \typo{the Yoshida} algorithm provides an approximation to the Langevin dynamics of weak order exactly one: approximating to high order the Hamiltonian dynamics does not imply approximating to high order the Langevin dynamics.
\rev{In addition, this also implies that higher  order algorithms for the Langevin equations are not required for better scaling, thus answering negatively one of the conjectures in the discussion of \cite{GC11}. }
All other choices of \(a\) and \(b\) in the family above lead to integrators with \(\nu=3\) and some of them show very good performance in the context of HMC algorithms \cite{BCSS14,CSASS21}. It may be proved \cite{BRSS18} that by increasing the number of gradient evaluations per proposal it is possible to have splitting integrators which attain arbitrarily high values of \(\nu\) and therefore arbitrarily better scaling properties of the sampler. (The relation between a target order \(\nu\) and the required number of gradient evaluations is known, but not simple.)

It is also possible to reach arbitrarily high values of \(\nu\) by means of \emph{implicit} integrators that generalize the implicit midpoint rule used in MAIMLA.
The Gauss method of \(s\) stages \cite{sanz2018numerical}, \(s = 1,2,\dots\), achieves \(\nu=2s+1\). The implicit midpoint corresponds to \(s=1\). These integrators have the property of exactly conserving energy for quadratic potentials and therefore possess an acceptance rate of \(100\%\) for Gaussian distributions. When applied to the potential \eqref{eq:target} in \(\bbR^m\) each proposal requires the solution of an \(s\times m\) dimensional system of nonlinear equations.

A comparison between the merits of the different proposals is \typo{beyond the scope of this paper}. However we would like to illustrate the theory with a numerical example. We have considered as a target a product of \(d\) independent \(N(0,1)\) univariate distributions,  \(d= 10^{3}, 10^{4}, 10^{5}\), and sampled, when the chain is started at stationarity, with MALA \(\delta_d = \ell^2/d^{1/3}\), fMALA \(\delta_d = \ell^2/d^{1/5}\) and the Yoshida algorithm presented above, \(\delta_d = \ell^2/d^{1/5}\). Recall that fMALA involves much more complexity than MALA and the splitting algorithm. For each sampler, simulations were carried out with a range of values of \(\ell\) and we recorded the empirical \typo{Squared Jumping Distance} and acceptance probability. \typo{The data reported are averages} over $10^{4}$ realizations.
Figure~\ref{figure} displays \typo{the Squared Jumping Distance} \(\times\: d^{1/5}\) as a function of acceptance rate. According to our analysis, the \(d^{1/5}\) scaling
 is appropriate for fMALA and Yoshida. We see that indeed for those two algorithms the (scaled) jumping distance does not degrade as \(d\) increases. However, the MALA results move downwards as \(d\) increases, because for this algorithm a scaling \(d^{1/3}\) would be required. The figure also displays, for each sampler, the relation \typo{in equation} \eqref{eq:EA} between \(\mathcal{SJD}_d \) and \rev{the acceptance rate}; since the prefactor \typo{\(2^{2/\nu}/(K\Sigma)^{1/\nu}\)} was not known (but could have been computed as  the target is Gaussian), we scaled the graphs by matching the maximum of the empirical results to equal the maximum of the corresponding theoretical curve. Finally, vertical lines
\typo{indicate} the theoretical values 0.574 and 0.704 of the optimal acceptance rate for methods with \(\nu = 3\) (MALA) and \(\nu = 5\) (fMALA and Yoshida). Clearly those values, proved to be valid in the limit \(d\rightarrow \infty\), agree very well with the simulations.\bigskip

{\color{black}{\bf Acknowledgement.} The authors are thankful to Prof.\ G. Zanella for bringing to  their attention the work by Vogrinc and Kendall. JMS has been supported by Ministerio de Ciencia e Innovación (Spain) through project PID2022-136585NB-C21, MCIN/AEI/10.13039/501100011033/\\FEDER, UE.
}


\bibliographystyle{plain}
\bibliography{references}

\end{document}